\documentclass[11pt,letterpaper]{article}

\usepackage{algorithm}
\usepackage{algorithmic}

\usepackage{latexsym}
\usepackage{times}
\usepackage{amsthm}
\usepackage{amsfonts}
\usepackage{amssymb}
\usepackage{paralist}
\usepackage{amsmath}
\usepackage[letterpaper,margin=2.3cm]{geometry}
\usepackage{appendix}
\usepackage[dvipsnames,usenames]{xcolor}
\usepackage{graphicx}
\usepackage{enumerate}
\usepackage{tabularx}
\usepackage[compact]{titlesec}
\usepackage[colorlinks,linkcolor=blue,filecolor=blue,citecolor=blue,urlcolor=blue,pdfstartview=FitH,plainpages=false]{hyperref}
\usepackage[capitalize]{cleveref}
\usepackage{subcaption}
\usepackage{caption}
\usepackage{tikz}
\usepackage[textsize=tiny]{todonotes}

\usepackage{xcolor}

\definecolor{darkred}{rgb}{0.75,0,0}

\usepackage{xspace}
\usepackage{lipsum}

\newcommand{\eps}{\varepsilon}

\usepackage{tabularx}
\usepackage{caption}
\usepackage{graphicx}
\usepackage{color}

\newtheorem{theorem}{Theorem}
\newtheorem{lemma}{Lemma}
\newtheorem{definition}{Definition}

\usepackage{framed}
\usepackage{sidecap}
\usepackage{enumerate}
\usepackage{amssymb}

\newcommand{\pd}{\mathbf{p}}

\newcommand{\med}{\text{mid}}

\def\polylog{\operatorname{polylog}}

\def\eps{{\epsilon}}

\newcommand{\ignore}[1]{}
\newcommand{\eat}[1]{}

\newcommand{\squishlist}{
 \begin{list}{$\bullet$}
  { \setlength{\itemsep}{0pt}
     \setlength{\parsep}{3pt}
     \setlength{\topsep}{3pt}
     \setlength{\partopsep}{0pt}
     \setlength{\leftmargin}{1.5em}
     \setlength{\labelwidth}{1em}
     \setlength{\labelsep}{0.5em} } }
\newcommand{\squishend}{
  \end{list}  }

\def\eps{{\epsilon}}

\newcommand{\anis}[1]{{\color{blue}\bf [Anisur: #1]}}

\begin{document}

\newpage  

\title{Local Mixing Time: \\ Distributed Computation and Applications}

\author{Anisur Rahaman Molla\thanks{Supported by DST Inspire Faculty research grant DST/INSPIRE/04/2015/002801.} \\
	{School of Computer Sciences}\\
	{NISER Bhubaneswar}\\
	{Odisha 752050, India}\\
	{\texttt{anisurpm@gmail.com}}
\and 
Gopal Pandurangan \thanks{Supported, in part, by NSF grants CCF-1527867, CCF-1540512,  IIS-1633720, CCF-1717075 and
BSF award 2016419.}\\
	{Department of Computer Science}\\
	{University of Houston}\\
	{Houston, Texas 77204, USA}\\
	 {\texttt{gopalpandurangan@gmail.com}}
}

\date{}

%

\maketitle


\begin{abstract} 
The mixing time of a graph is an important metric, which is not only useful in analyzing connectivity and expansion properties of the network, but also serves as a key parameter in designing efficient algorithms. We introduce a new notion of mixing of a random walk on a (undirected) graph, called {\em local mixing}. Informally, the local mixing with respect to a
given node $s$, is the mixing of a random walk probability distribution restricted to a large enough subset of nodes --- say, a subset of  size at least $n/\beta$ for a given parameter $\beta$ --- containing $s$. The time to mix over such a  subset  by a random walk starting from a source node $s$ is called the {\em local mixing time} with respect to $s$.  The local mixing time captures the local connectivity and expansion properties around a given source node and  is  a useful parameter  that determines the running time of algorithms for partial information spreading, gossip etc.   

Our first contribution is formally defining the notion of {\em local mixing time} in an undirected  graph. We then present an efficient distributed algorithm which computes a constant factor approximation to the local mixing time with respect to  a source node $s$ in $\tilde{O}(\tau_s)$ rounds\footnote{The notation
$\tilde{O}$ hides a $O(\polylog n )$ factor.}, where $\tau_s$ is the local mixing time w.r.t $s$ in an $n$-node regular graph. This bound holds when $\tau_s$ is significantly smaller than the  
conductance of the local mixing set (i.e., the set where the walk mixes locally); this is typically the interesting case where the local mixing time
is significantly smaller than the mixing time (with respect to  $s$).  We also present a distributed algorithm that computes the {\em exact} local mixing time in $\tilde{O}(\tau_s \mathcal{D})$ rounds, where $\mathcal{D} =\min\{\tau_s, D\}$ and $D$ is the diameter of the graph (this bound holds unconditionally without any assumptions on $\tau_s$).
Our algorithms work in the {\em CONGEST} model of distributed computing. Since the local mixing time can be significantly smaller than
the mixing time  (or even the diameter) in many graphs, it serves as a tighter measure of  distributed complexity in certain algorithmic applications. In particular, we show that local mixing time tightly characterizes the complexity of partial information spreading which in turn is useful in solving other problems such as the maximum coverage problem, full information spreading, leader election etc. 

 
\end{abstract}
\vspace{.2in}

\newpage
\section{Introduction}\label{sec:intro}

Mixing time of a random walk in a graph is the time taken  by a  random walk to converge to the {\em stationary distribution} of the underlying graph.  It is an important parameter which is closely related to various key graph properties such as graph expansion, spectral gap, conductance etc. Mixing time (denoted by $\tau^{mix}$) is related to the {\em conductance} $\Phi$ and {\em spectral gap} ($1-\lambda_2$) of a $n$-node graph due to the known relations (\cite{JS89}) that $\frac{1}{1-\lambda_2}\leq \tau^{mix} \leq \frac{\log n}{1-\lambda_2}$ and $\Theta(1-\lambda_2)\leq \Phi \leq \Theta(\sqrt{1-\lambda_2})$, where  $\lambda_2$ is the second largest eigenvalue of the adjacency matrix of the graph.   Small mixing time means the graph has  high expansion and spectral gap. Such a network supports fast random sampling (which has many applications \cite{drw-jacm}) and low-congestion routing \cite{mihail}. Moreover, the spectral properties tell a great deal about the network structure \cite{DasSarmaMPU15}. Mixing time is also useful in designing efficient randomized algorithms in communication networks \cite{storage-spaa13,APR-podc13,SarmaMP15,DasSarmaMPU15,KM15,sirocco14}.

There has been some previous work on distributed algorithms to compute mixing time.
The work of Kempe and McSherry \cite{kempe}  estimates the mixing time $\tau^{mix}$ in $O(\tau^{mix} \log^2 n)$ rounds. Their approach uses {\em Orthogonal Iteration} i.e., heavy matrix-vector multiplication process, where each node needs to perform complex calculations and do memory-intensive  computations. This may not be suitable in
a lightweight environment.
 It is mentioned in their paper that it would be interesting whether a simpler and direct approach based on eigenvalues/eigenvectors can be used to compute mixing time.  Das Sarma et al. \cite{drw-jacm} presented a distributed algorithm based on sampling nodes by performing sub-linear time random walks and then comparing the distribution with stationary distribution.  The work of Molla and Pandurangan \cite{icdcn17} presented  an efficient and simple distributed algorithm for computing the mixing time of  undirected graphs. Their algorithm  estimates the mixing time $\tau^{mix}_s$ (with respect to a source node $s$) of any $n$-node undirected graph in $O(\tau^{mix}_s \log n)$ rounds  and achieves high accuracy of estimation.
This algorithm is based on random walks and require very little memory and use lightweight local computations, and works in the {\em CONGEST} model. The algorithm of Das Sarma et al.  can be sometimes faster than the algorithm of Molla and Pandurangan, however, there is a grey area (in the comparison between the two distributions) for which the former algorithm fails to estimate the mixing time with any good accuracy (captured
 by the accuracy parameter $\eps$ defined in
 Section \ref{sec:rwalk}). The latter algorithm  is sometimes faster (when the mixing time is $o(\sqrt{n})$) and estimates the mixing time  with high accuracy \cite{icdcn17}. 

In this paper,
we introduce a new notion of mixing (time) of a random walk on a (undirected) graph, called {\em local mixing (time)}. 
Local mixing time  (precisely defined in Definition \ref{def:loc-mix-time})   captures the local connectivity and expansion properties around a given source node and  is  a useful parameter  that determines the run time of algorithms for information spreading, gossip etc. Informally,   {\em local mixing time}   is the time
for a token from any node $s$ to reach (essentially) the stationary distribution of a large enough subset $S$ of nodes (say of size at least $n/\beta$, for a given parameter $\beta$) containing $s$  (here, stationary distribution is computed with respect to that subset).
(It is important to note that the set $S$ is not known a priori, it just needs to exist.)  Local mixing time is a finer notion than mixing time and is always
upper bounded by mixing time (trivially), but can be significantly smaller than the mixing time (and even the diameter) in many graphs (cf. Section \ref{sec:graph-examples}).
For example, the mixing time of a $\beta$-barbell graph (cf. Section \ref{sec:graph-examples}) is $\Omega(n)$ (and its diameter is $O(\beta)$), whereas its local mixing time is $O(1)$; hence partial information spreading (cf. Section \ref{sec:app}) is significantly faster in such graphs. 

Our main contribution is an efficient distributed algorithm for computing the local mixing time in undirected regular graphs. We show that we can compute a constant factor approximation (for any small fixed positive constant)\footnote{We actually compute
a 2-factor approximation, but it can be easily modified to compute any $(1+\delta)$-factor approximation, for any constant $\delta > 0$.} of
local mixing time in $O(\tau_s \log^2 n \log_{(1+\eps)} \beta)$ rounds in undirected  graphs, where $\tau_s$ is the local mixing time\footnote{Please see Section \ref{sec:loc-mix-def} for the formal definition and notation; we formally denote the local mixing time by  $\tau_s(\beta, \eps)$ parameterized
by $\beta$ (which determines the size of the set where the walk locally mixes) and by $\epsilon$, an accuracy parameter (which measures how close the walk mixes).} with respect
to $s$. This bound holds when $O(\tau_s)$ is significantly smaller than the  
conductance of the local mixing set (i.e., the set where the walk mixes locally); this is typically the interesting case where the local mixing time
is significantly smaller than the mixing time of $s$.  We also present a distributed algorithm that computes the {\em exact} local mixing time in $O(\tau_s \mathcal{D} \log n \log_{(1+\eps)} \beta)$ rounds, where $\mathcal{D} = \min\{\tau_s, D\}$. This bound holds unconditionally without any assumptions on $\tau_s$. The local mixing time of the graph is the maximum 
of the local mixing times with respect to every node in the graph. We note that one can compute the local mixing time
with respect to the entire graph by taking the maximum of all the local mixing times starting from each vertex. This
(in general) will incur an $O(n)$-factor additional overhead on the number of rounds (by running the distributed algorithm with respect to every node). However, depending on the input graph, one may be able to compute (or approximate) it significantly faster  by sampling only a few source nodes and running it only from those source nodes (e.g., in a graph where the local mixing times are more or less the same with respect to any node).

Our definition of local mixing time is inspired by the notion of {\em weak conductance} \cite{censor-podc10} that similarly tries
to capture the conductance around a given source vertex. It was shown in \cite{censor-podc10} that
weak conductance captures the performance of  {\em partial information spreading}. In partial information spreading, given a $n$-node graph with each node having  a (distinct) message, the goal is to disseminate  each  message to a fraction of the total number of nodes --- say $n/c$, for some $c>1$ --- and to ensure that each node receives at least $n/c$ messages.  It was shown
that graphs which have {\em large} weak conductance (say a constant) admit efficient information spreading, despite having
a poor (small) conductance \cite{censor-podc10}; hence weak conductance better captures the performance of partial information spreading.
While it is not clear how to compute weak conductance efficiently, we show that local mixing time also captures
partial information spreading. In Section \ref{sec:app}, we show that the well-studied ``push-pull" mechanism   achieves partial information spreading
in $O(\tau \log n)$ rounds, where $\tau$ is local mixing time
with respect to the entire graph, i.e., $\tau = \max_{s\in V} \tau_s$. As shown in \cite{censor-podc10}, an application of partial information spreading
is to the {\em maximum coverage problem} which naturally arises in circuit layout, job scheduling and facility location, as well as in distributed resource allocation with a global budget constraint.




 Our algorithms work in CONGEST model of distributed computation where only small-sized messages ($O(\log n)$-sized messages) are allowed in every communication round between nodes. Moreover, our algorithms are simple, lightweight 
 (low-cost  computations within a node) and easy to implement.  
We note that our bounds are non-trivial in the CONGEST model.\footnote{In the LOCAL model, all problems can be trivially solved in $O(D)$ rounds by collecting all the topological information at one node, whereas in the CONGEST model, the same will take $O(m)$ rounds, where $m$ is the number of edges in the graph.} In particular, we point out that one cannot obtain
these bounds by simply extending the algorithm of \cite{icdcn17} that computes the mixing time $\tau^{mix}_s$ (with respect to a source node $s$) of any $n$-node undirected graph in $O(\tau^{mix}_s \log n)$ rounds.  Informally, the main difficulty in computing
(or estimating) the local mixing time is that one does not (a priori) know the set where the walk locally mixes (there can
be exponential number of such sets). This calls for a more sophisticated approach, yet we obtain a bound that is comparable to
the bound obtained for computing the mixing time obtained in \cite{icdcn17}.

\subsection{Distributed Network Model}
We model the communication network as an undirected, unweighted, connected graph $G = (V, E)$, where $|V| = n$ and $|E| = m$. Every  node has limited initial knowledge. Specifically, we assume that each node is associated with a distinct identity number  (e.g., its IP address). 
At the beginning of the computation, each node $v$ accepts as input its own identity number and the identity numbers of its neighbors in $G$.
We also assume that the number of nodes and edges i.e., $n$ and $m$ (respectively) are given as inputs. (In any case, nodes can compute them easily through broadcast in $O(D)$, where $D$ is the network diameter.) The nodes are only allowed to communicate through the edges of the graph $G$. We assume that the communication occurs in  synchronous  {\em rounds}. 
We will use only small-sized messages. In particular, in each round, each node $v$ is allowed to send a message of size $O(\log n)$ bits through each edge $e = (v, u)$ that is adjacent to $v$.  The message  will arrive to $u$ at the end of the current round. 
This is a  widely used  standard model known as the {\em CONGEST model} to study distributed algorithms (e.g., see \cite{peleg,gopal-book}) and captures the bandwidth constraints inherent in real-world computer  networks. 

We  focus on minimizing the  {\em the running time}, i.e., the number of {\em rounds} of distributed communication. Note that the computation that is performed by the nodes locally is ``free'', i.e., it does not affect the number of rounds; however, we will only perform polynomial cost computation locally (in particular, very simple computations) at any node. 

 For any node $u$, $d(u)$ and $N(u)$ denote the degree of $u$ and the set of neighbors of $v$ in $G$ respectively. 
\subsection{Related Work}
\label{sec:related}
We briefly discuss prior work that relates to the related problem of computing the mixing time of a graph. It is important
to note that these algorithms do not give (or cannot be easily adapted) to give efficient algorithms for computing the local
mixing time.

Das Sarma et al. \cite{drw-jacm} presented a fast decentralized algorithm for estimating mixing time, conductance and spectral gap of the network. In
particular, they show that given a starting node $s$, the mixing time with respect to $s$, i.e, $\tau^{mix}_s$, can be
estimated in $\tilde{O}(n^{1/2} + n^{1/4}\sqrt{D\tau^{mix}_s})$ rounds. This gives an alternative algorithm to the only previously known
approach by Kempe and McSherry \cite{kempe} that can be used to estimate
$\tau^{mix}_s$ in $\tilde{O}(\tau^{mix}_s)$ rounds. 
 In fact, the work of \cite{kempe} does more and gives a decentralized algorithm for
computing the top $k$ eigenvectors of a weighted adjacency matrix
that runs in $O(\tau^{mix}\log^2 n)$ rounds if two adjacent nodes are allowed to exchange $O(k^3)$ messages per round, where $\tau^{mix}$ is
the mixing time and $n$ is the size of the network.  

Molla and Pandurangan \cite{icdcn17}  presented an  algorithm that estimates the mixing time $\tau^{mix}_s$ for the source node in $O(\tau^{mix}_s \log n)$ rounds in a undirected graph and achieves high accuracy of estimation. 
This algorithm is based on random walks. Their approach, on a high-level, is based on efficiently performing many random walks from a particular node and computing the fraction of random walks that terminate over each node. They show that this fraction estimates the random walk probability distribution. This approach achieves very high accuracy which is a requirement in some applications \cite{DasSarmaGP09,SarmaMP15,KM15,SpielmanT04}. As mentioned earlier, this approach does not extend
to computing the local mixing time efficiently.

The algorithm of Das Sarma et al. 
\cite{drw-jacm} is based on sampling nodes by performing sub-linear time random walks of certain length and comparing the distribution with the stationary distribution. In particular, if $\tau^{mix}$ is smaller than $\max \{\sqrt{n}, n^{1/4} \sqrt{D}\}$, then
 the algorithm of Molla and Pandurangan is faster.  Also there is a grey area for the accuracy parameter $\eps$ for which the algorithm of Das Sarma et al. cannot estimate the mixing time. 
 More precisely, the algorithm
of Das Sarma et al. estimates the mixing time for accuracy parameter $\eps = 1/(2e)$ with respect to a source node $s$, $\tau^{mix}_s(1/2e)$  as follows: the estimated value will be between the true value  and $\tau^{mix}_s(O(1/(\sqrt{n}\log n)))$.



The notion of weak conductance was defined in the work of Censor-Hillel and Sachnai \cite{censor-podc10} which they then use
as a parameter to capture partial information spreading.  They also showed that partial information spreading is useful in solving several other important problems, e.g., maximum coverage, full information spreading, leader election  etc. \cite{censor-podc10, censor-soda11}.

There are some notions proposed in the literature that are alternative to the standard notion of mixing time and stationary distribution. 
These notions are different from the notion of local mixing time  studied in this paper.
The work of \cite{AFPP-soda12} introduces the concept of ``metastable" distribution and pseudo-mixing time of Markov chains.
Informally, a distribution $\mu$ is $(\epsilon, T)$- metastable for a Markov chain if, starting from $\mu$, the Markov chain stays at distance at most $\epsilon$ from $\mu$ for at least $T$ steps. The pseudo-mixing time of $\mu$ starting from a state $x$ is the number of steps needed by the Markov chain to get $\epsilon$-close to $\mu$ when started from $x$. Another notion that has been studied in literature  is ``quasi-stationarity", which  has been used to model the long-term behaviour of stochastic systems that appear to be stationary over a reasonable time period, see, e.g., \cite{DM15} for more details.

\section{Local Mixing}\label{sec:rwalk}
We define the notion of {\em local mixing} and {\em local mixing time}. Before we do that, we first recall some preliminaries
on random walks.

\subsection{Random Walk Preliminaries}
Given an undirected graph $G$ and a starting point, a {\em simple random walk} is defined as: in each step, the walk goes from the current node to a random neighbor i.e., from the current node $u$, the probability of moving to node $v$ is $\Pr(u, v) = 1/d(u)$ if $(v,u) \in E$, otherwise $\Pr(u, v) = 0$, where $d(u)$ is the degree of $u$.

Suppose a random walk starts at vertex $s$. Let $\pd_0 (s)$ be the initial distribution with probability $1$ at the node $s$ and zero at all other nodes. Then the probability distribution $\pd_t (s)$ at time $t$ starting from the initial distribution $\pd_0 (s)$ can be seen as the matrix-vector multiplication $A^t\pd_0(s)$, where $A$ is the transpose of the transition probability matrix of $G$. We denote the probability distribution vector at time $t$ by the bold letter $\pd_t (s)$ and the probability of a co-ordinate i.e., probability at a node $v$ by $p_t(s, v)$. Sometime we omit the source node from the notations, when it is clear from the text--- so the notations would be $\pd_t$ and $p_t(v)$ respectively. The stationary distribution (a.k.a steady-state distribution) is the distribution $\pd_r$ such that $A\pd_r = \pd_r$ i.e.,  the distribution doesn't change (it has converged). The stationary distribution of an undirected connected graph is a well-defined quantity which is $\bigl(\frac{d(v_1)}{2m}, \frac{d(v_2)}{2m}, \ldots, \frac{d(v_n)}{2m}\bigr)$, where $d(v_i)$ is the degree of node $v_i$.  
We denote the stationary distribution vector by $\pmb{\pi}$, i.e., $\pi(v) = d(v)/2m$ for each node $v$. The stationary distribution of a graph is fixed irrespective of the starting node of a random walk, however, the number of steps (i.e., time) to reach to the stationary distribution could be different for different starting nodes. The time to reach to the stationary distribution is called the {\em mixing time} of a random walk with respect to the source node $s$. The mixing time corresponding to the source node $s$ is denoted by $\tau^{mix}_s$. The mixing time of the graph, denoted by $\tau^{mix}$, is the maximum mixing time among all (starting) nodes in the graph.  Mixing time exists and is well-defined for non-bipartite graphs; throughout we assume non-bipartite graphs.\footnote{Bipartiteness or not is rather a technical issue, since  if we consider a lazy random walk (i.e.,
random walk where at each step, with probability $1/2$ the walk stays in the same node and with probability $1/2$, it goes to a random neighbor), then it is well-defined for all graphs.}   The formal definitions are given below. 


\begin{definition}\label{def:mixing-time} ($\tau^{mix}_s(\eps)$--mixing time for source $s$ and $\tau^{mix}(\eps)$--mixing time of the graph)\\
Define $\tau^{mix}_s (\eps)= \min \{t : ||\pd_t - \pmb{\pi}||_1 < \eps\}$, where $||\cdot||_1$ is the $L_1$ norm. Then $\tau^{mix}_s(\eps)$ is called the $\eps$-near mixing time for any $\eps$ in $(0, 1)$. The mixing time of the graph is denoted by $\tau^{mix} (\eps)$ and is defined by $\tau^{mix}(\eps) = \max \{\tau^{mix}_v(\eps): v \in V\}$. It is clear that $\tau^{mix}_s (\eps) \leq \tau^{mix} (\eps)$. \hfill $\square$
\end{definition}

We sometime omit $\eps$ from the notations when it is understood from the context. The definition of $\tau^{mix}_s$ is consistent due to the following standard monotonicity property of distributions. We note that a similar monotonicity property
does not hold for $\tau_s$, the local mixing time with respect to source node $s$; this is one reason why computing local mixing time is more non-trivial compared to mixing time.

\begin{lemma}\label{lem:monotonicity}
$||\pd_{t+1} - \pmb{\pi}||_1 \leq  ||\pd_t - \pmb{\pi}||_1$
\end{lemma}

\begin{proof}(adapted from Exercise 4.3 in \cite{Levin})
The monotonicity follows from the fact that $||A\mathbf{x}||_1 \le ||\mathbf{x}||_1$, where $A$ is the transpose of the transition probability matrix of the graph and $\mathbf{x}$ is any $n\times 1$ vector. That is, $A(i,j)$ denotes the probability of transitioning from the node $j$ to the node $i$. This in turn follows from the fact that the sum of entries of any column of $A$ is 1.

We know that $\pmb{\pi}$ is the stationary distribution of the transition matrix $A$. This implies that if $\ell$ is $\eps$-near mixing time, then $||A^{\ell}\pd_0 - \pmb{\pi}||_1 \leq \eps$, by definition of $\eps$-near mixing time and $\pd_{\ell} = A^{\ell}\pd_0$. Now consider $||A^{\ell+1}\pd_0 - \pmb{\pi}||_1$. This is equal to $||A^{\ell+1}\pd_0 - A\pmb{\pi}||_1$, since $A\pmb{\pi} = \pmb{\pi}$.  However, this reduces to $||A(A^{\ell}\pd_0 - \pmb{\pi})||_1 \leq ||A^{\ell}\pd_0 - \pmb{\pi}||_1 \leq \eps$, (from the fact  $||A\pd||_1 \le ||\pd||_1$). Hence, it follows that $(\ell+1)$ is also $\eps$-near mixing time.
\end{proof}

\subsection{Definition of Local Mixing and  Local Mixing Time}\label{sec:loc-mix-def} 

For any set $S\subseteq V$, we define $\mu(S)$ is the volume of $S$ i.e., $\mu(S) = \sum_{v \in S} d(v)$. Therefore, $\mu(V) = 2m$ is the volume of the vertex set. The {\em conductance} of the set $S$ is denoted by $\phi(S)$ and defined by 
$$\phi(S) =  \frac{|E(S, V\setminus S)|}{\min\{\mu(S),\, \mu(V\setminus S)\}},$$
where $E(S, V\setminus S)$ is the set of edges between $S$ and $V\setminus S$. 

Let us define a vector $\pmb{\pi}_{S}$ over the set of vertices $S$ as follows: 

\[ \pi_{S}(v) =
  \begin{cases}
    d(v)/\mu(S)       & \quad \text{if } v \in S\\
    0  & \quad \text{otherwise} \\
  \end{cases}
\]

Notice that $\pmb{\pi}_V$ is the stationary distribution $\pmb{\pi}$ of a random walk over the graph $G$, and $\pmb{\pi}_S$ is the restriction of the distribution on the subgraph induced by the set $S$. Recall that we defined $\pd_t$ as the probability distribution over $V$ of a random walk of length $t$, starting from some source vertex $s$. Let us denote the restriction of the distribution $\pd_t$ over a subset $S$ by $\pd_t {\restriction_S}$ and define it as: 

\[ p_t {\restriction_S} (v) =
  \begin{cases}
    p_t(v)       & \quad \text{if } v \in S\\
    0  & \quad \text{otherwise} \\
  \end{cases}
\] 
It is clear that $p_t {\restriction_S}$ is not a probability distribution over the set $S$ as the sum could be less than $1$. 

Informally, {\em local mixing}, with respect to a source node $s$, means that there exists  some (large-enough) subset of nodes $S$ containing $s$ such that the random walk probability distribution becomes close to the stationary distribution restricted to $S$ (as defined above) quickly. We would like to quantify how fast the walk mixes locally around a source vertex. This is called as {\em local mixing time} which is formally defined below.    

\begin{definition}\label{def:loc-mix-time} (Local Mixing and Local Mixing Time) \\
Consider a vertex $s \in V$. Let $\beta \geq 1$ be a positive constant and $\epsilon \in (0,1)$ be a fixed parameter. We  first
 define the notion of local mixing in a set $S$.  Let $S  \subseteq V$ be a {\em fixed} subset containing $s$ of size at least $n/\beta$. Let $\pd_t{\restriction_S}$ be the restricted probability distribution over $S$ after $t$ steps of a random walk starting from $s$ and  $\pmb{\pi}_S$ be as defined above.  Define the {\em  mixing time  with respect to set $S$}  as  $\tau^S_{s}(\beta, \eps) = \min \{t: ||\pd_t{\restriction_S} - \pmb{\pi}_S ||_1 < \eps \}$. We say that the random walk {\em locally mixes} in $S$ if  $\tau^S_{s}(\beta, \eps)$ exists and well-defined. (Note that a walk may not locally mix in a given set $S$, i.e., there exists no time $t$ such that  $||\pd_t{\restriction_S} - \pmb{\pi}_S ||_1 < \eps$; in this case we can take the local mixing time to be $\infty$.)

The local mixing time with respect to source node $s$ is defined as
 $\tau_{s}(\beta, \eps) = \min_{S}\tau^S_{s}(\beta, \eps)$, where the minimum is taken over all subsets $S$ (containing $s$) of size
at least $n/\beta$, where the random walk starting from $s$ locally mixes. A set $S$  where the minimum  is attained (there may be more than one) is
called the {\em local mixing set}.
The local mixing time of the graph, $\tau(\beta, \eps)$ (for  given parameters $\beta$ and $\epsilon$), is $\max_{v \in V} \tau_{v}(\beta, \eps)$. \hfill $\square$
\end{definition}

From the above definition, it is clear that $\tau_{s}(\beta, \eps)$ always exists (and well-defined) for every fixed $\beta \geq 1 $, since in the worst-case,
it equals the mixing time of the graph; this happens when $|S| = n \geq n/\beta$ (for every $\beta \geq 1$). We note that, crucially, in the above definition of local mixing time,  the {\em minimum} is taken over subsets $S$ of size at least $n/\beta$, and thus, in many graphs, 
local mixing time can be substantially smaller than the mixing time when $\beta > 1$ (i.e., the local mixing can happen much earlier in
some set $S$ of size $\geq n/\beta$ than
the mixing time). It is important to note that the set $S$ where the  local mixing time is attained is
not fixed  a priori, it only requires that a set $S$ of size at least $n/\beta$ exists. (Since $S$ is not known a priori, the computation of local mixing time is more complicated, unlike mixing time; in our algorithms we do not explicitly compute the local mixing set, but 
only compute an approximation of the local mixing time). 

It also follows from the definition that the local mixing time depends on the parameter $\beta$, i.e., size of subset $S$ --- in general,
smaller the size of $S$ smaller the local mixing time. In particular, if $\beta = 1$, then $\tau_{s}(1, \eps) = \tau^{mix}_{s}(\eps)$, mixing time for source $s$ (cf. Definition~\ref{def:mixing-time}) and in general, $\tau(\beta, \eps) \leq \tau^{mix}(\eps)$ for any $\beta$. 

Intuitively, small local mixing time implicates that the random walk starting from a vertex mixes fast over a (large enough) subset (parameterized by $\beta$)  around that vertex. Therefore, given an undirected graph $G$, a source node $s$ and a parameter $\beta$, the goal is to compute the local mixing time with respect to $s$.\footnote{Similar to the case of the mixing time, one can compute the  local mixing time
with respect to the entire graph by taking the maximum of all the local mixing times starting from each vertex. This
(in general) will incur an $O(n)$-factor additional overhead on the number of rounds.} In our algorithm in Section \ref{sec:local-mixing-time}, we compute a constant factor approximation to the local mixing time (we do not explicitly compute the set where the walk locally mixes). In Section \ref{sec:general-graph-local-mixing-time}, we give an algorithm to compute the exact local mixing time.

\subsection{Local Mixing Time and Mixing Time in Some Graphs}\label{sec:graph-examples}
The local mixing time $\tau_{s}(\beta, \eps)$ w.r.t. a source node $s$ (also $\tau(\beta, \eps)$) is monotonically decreasing function of $\beta$. That is if $\beta_1 \geq \beta_2$ then $\tau_{s}(\beta_1, \eps) \leq \tau_{s}(\beta_2, \eps)$ (also $\tau(\beta_1, \eps) \leq \tau(\beta_2, \eps)$). This follows directly from the definition since $n/\beta_1 \leq n/\beta_2$. 

Let us now compare the local mixing time and mixing time in some well-known graph classes. It will strengthen understanding towards why local mixing time is a refined measure of mixing time of a random walk in a graph. Consider the following graphs: 
\begin{enumerate}[(a)]
 \item {\em Complete graph:}  Both local mixing time and mixing time are constant. This is because, in one step of the random walk, the probability distribution becomes $\pd_{1} = (0, \frac{1}{n-1}, \frac{1}{n-1}, \ldots, \frac{1}{n-1})$ which is $\eps$-close to the uniform distribution (which is the stationary distribution). Thus the mixing time of the complete graph is $1$ and hence, the local mixing time is equal to the mixing time.    
 \item {\em $d$-regular expander:} It is known that the mixing time of an expander graph is $O(\log n)$ \cite{Levin}. The proof follows from the expansion property of the graph. The rate of convergence of a probability distribution to the stationary distribution is bounded by the second largest eigenvalue of the transition matrix. The second largest eigenvalue of an expander graph is constant. It can be shown that mixing in a set of size at least $n/\beta$, will take at least $O(\log_d (n/\beta)) = O(\log n)$ time (for constant $d$ and 
 $\beta$). Thus the local mixing time is $O(\log n)$. Therefore, there is no substantial difference between  mixing time and local mixing time in expander graphs.
 \item {\em Path:} It is known that mixing time of a path of $n$ nodes is $O(n^2)$ \cite{Levin}.  The local mixing time is $O(n^2/\beta^2)$, since it requires so much time to mix in a sub-path of size $n/\beta$. This can be substantially smaller than mixing time when $\beta$ is large. 
 \item {\em $\beta$-barbell graph:}  This is a generalization of the {\em barbell} graph. The $\beta$-barbell graph consists of a path of $\beta$ equal sized cliques, i.e., the size of each clique is $n/\beta$ (see  Figure~\ref{fig:barbell}). The local mixing time is $1$, but it is easy to show that mixing time is $\Omega(\beta^2)$.
In this graph, there is a siginificant difference between mixing time and local mixing time, e.g., for $\beta = \sqrt{n}$, the difference between mixing time and local mixing time is $O(n)$. Similar graph structures (e.g., class of graphs with $\beta$ equal-sized connected components, which have very small mixing time such as expanders, that are connected via a path or ring) have a large gap between mixing time and local mixing time. 
\end{enumerate}

 \begin{figure}
\centering
    \includegraphics[scale=.5]{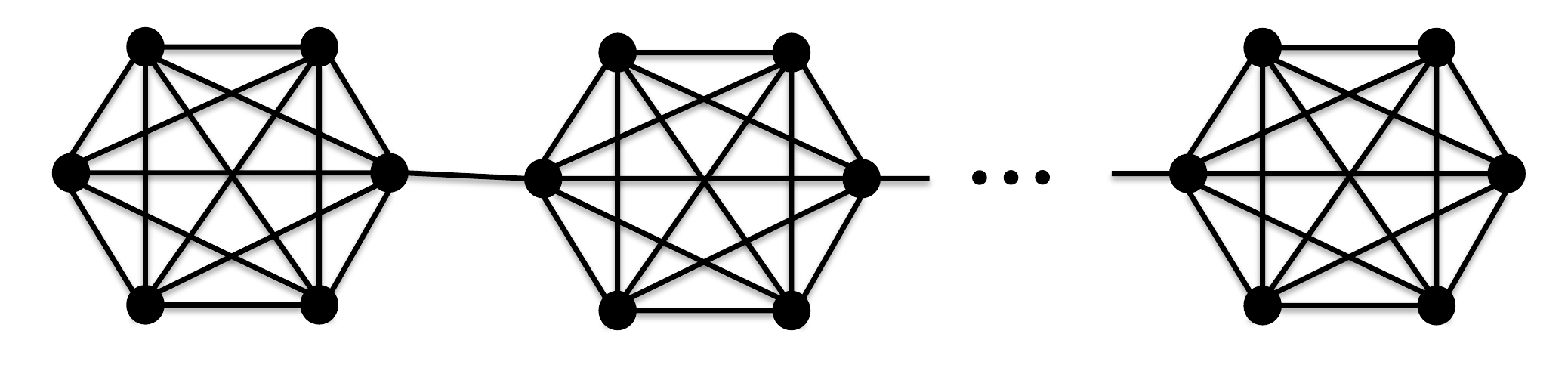}
\caption{$\beta$-barbell graph: a path of $\beta$ cliques of equal size. }
\label{fig:barbell}
\end{figure}
We next present a deterministic approach to compute the probability distribution $\pd_{\ell}$ of a random walk of any length $\ell$. The idea is adapted from the paper \cite{KM15} and explored in this paper to compute local mixing time.    

\subsection{Computation of Random Walk Probability Distribution}\label{sec:deterministic-simulation}

Let us compute the probability distribution $\pd_{\ell}$ starting from a given source node in the graph $G$. We present an algorithm (Algorithm~\ref{alg:rw-probability}) which approximates $\pd_{\ell}$ in $\ell$ time in the CONGEST model. The algorithm essentially simulates the probability distribution of each step of the random walk starting from a source node by a deterministic flooding technique. At the beginning of any round $t$, each node $u$ sends $p_{t-1}(s, u)/d(u)$ to its $d(u)$-neighbors and at the end of the round $t$, each node $u$ computes $p_t(s, u) = \sum_{v \in N(u)} p_{t-1}(s, v)/d(v)$. After $\ell$ rounds, each node $u$ will output its (estimated) probability $\tilde{p}_{\ell}(u)$. The estimated probabilities $\tilde{\pd}_{\ell}$ can be made as close as to the exact values $\pd_{\ell}$, i.e., $\mid \tilde{p}_{\ell}(u) - p_{\ell}(u) \mid < \eps$, for any small $\eps < 1$. In fact, this deterministic approach can compute exact probability distribution $\pd_{\ell}$ in principle. Since, in the CONGEST model, only $O(\log n)$ bits are allowed to be exchanged, it's not possible to send a real number $p_t(s, u)$ through an edge; instead an approximated value (rounding off) of size $O(\log n)$ bits can be sent. Thus, it is possible to compute a close approximation to the probability distribution $\pd_{\ell}$ of a random walk of any length $\ell$.   


\begin{algorithm}[H] 
\small
\caption{\sc Estimate-RW-Probability}
\label{alg:rw-probability}
\textbf{Input:} A graph $G = (V, E)$, a source node $s$ and the length $\ell$.\\
\textbf{Output:} Each node $u$ outputs $\tilde{p}_{\ell}(u)$. 

\begin{algorithmic}[1]
\STATE Initialization: at source node $s$, $w_0(s) = 1$ and at all other nodes $u$, $w_0(u) = 0$. 
 
\FOR{each round $t = 1, 2, \dots, \ell$}
\STATE Each node $u$ whose $w_{t-1}(u) \neq 0$, does the following in parallel: \label{stp:det-flooding}\\
 (i) send $w_{t-1}(u)/d(u)$ to all the neighbors $v \in N(u)$. \\
 (ii) Compute the sum (say, $\sigma$) of the received values from all neighbors $v \in N(u)$  and round it to the closest integer $nint(\sigma n^c)$ multiple of $1/n^c$, for any integer $c \geq 6$, where $nint(\cdot)$ is nearest integer function. Store this rounded value as $w_t(u)$. 
\ENDFOR  

\STATE Each node $u$ outputs $\tilde{p}_{\ell}(u) = w_{\ell}(u)$. 
\end{algorithmic}
\end{algorithm}

Note that at each step the value $\sum_{v \in N(u)} w_{t-1}(v)/d(v)$ at node $u$ is rounded to the closest integer multiple of $1/n^c$. Intuitively, the error of estimation is at most $1/n^c$ for each step. Thus the following error bound (Lemma~\ref{lem:error-bound}) of the approximation holds. The proof can be easily adapted from the Lemma~8 in \cite{KM15}.
\begin{lemma}\label{lem:error-bound}
At any time $t$,  $\mid \tilde{p}_{t}(u) - p_{t}(u) \mid < tn^{-c}$, for all the nodes $u$. 
\end{lemma}

Therefore, the algorithm finishes in $\ell$ time and computes a close approximation of the probabilities $p_{\ell}(u)$. Since the mixing time (and hence the local mixing time) is at most $O(n^3)$ for any graph, choosing $c = 6$ would suffice to get a very small approximation error. It is to be noted that a randomized algorithm presented in \cite{icdcn17} does the same job with high probability in $\ell$ time as well.


\section{Local Mixing Time Computation}\label{sec:local-mixing-time}

Let us assume the graph is regular and degree of each node is $d$. 
Then the volume of any set $S\subseteq V$ is $\mu(S) = d|S|$ and the non-zero entries in the restricted stationary distribution $\pi_S$ are all $1/|S|$. Let $s$ be the given source node from where the local mixing time $\tau_{s}(\beta, \eps)$ needs to be computed. We assume the error of estimation $\eps$ to be any  arbitrarily small (but fixed) positive constant  in the Definition~\ref{def:loc-mix-time} (say, we can choose $\eps = 1/8e$ which is typically done). 
Further, we assume that the graph satisfies the condition $\tau_s(\beta, \eps)\phi(S) = o(1)$ for every  $s$ and for every set $S$,  where $S$ is the set  where the random walk locally mixes (cf. Definition \ref{def:loc-mix-time}). (Note that we don't know 
$S$ a priori). We make this assumption so that our algorithm can compute a $2$-approximation of the local mixing time $\tau_s(\beta, \eps)$ efficiently; this is typically the interesting case, when the local mixing time is much smaller than the mixing time. We also show an easy extension of the algorithm to compute the (exact) local mixing time $\tau_s(\beta, \eps)$ in general regular graphs (without any conditions), but that takes slightly longer time. Therefore, the goal is to compute the minimum time $t$, such that $||\pd_t{\restriction_S} - 1/|S| ||_1 < \eps$, on a set $S$ that is as small as possible, but of size at least $n/\beta$. Recall that $||\pd_t{\restriction_S} - 1/|S| ||_1 = \sum_{u\in S} |p_t(u) - 1/|S||$. 

\begin{algorithm}[h] 
\small
\caption{\sc Local-Mixing-Time}
\label{alg:local-mixing-time}
\textbf{Input:} A graph $G = (V, E)$, a source node $s$, a positive constant $\beta$ and a fixed accuracy parameter $\eps$ (arbitrarily small positive constant).\\
\textbf{Output:} An approximate local mixing time $\tau(\beta, \eps)$. 

\begin{algorithmic}[1]
 
\FOR{each $h = 0, 1, 2, 3,  \dots $}
\STATE $\ell = 2^h$ \label{stp:creat-token}

\STATE The node $s$ computes a BFS tree of depth $\min\{D, \ell \}$ via flooding. \label{stp:bfs-tree}

\STATE Run Algorithm~\ref{alg:rw-probability} with $s$ as source node and $\ell$ as the length. Each node $u$ will have $p_{\ell}(u)$ in the end. \label{stp:rw-prob-computation}

\FOR{$R = n/\beta, (1 +\eps)n/\beta, (1 + \eps)^2n/\beta, \ldots, n$} \label{stp:doubling-set}

\STATE Each node $u$ computes the difference $x_u = |p_{\ell}(u) - 1/R|$. 
 
\STATE Node $s$ computes the sum of $R$ smallest $x_u$ values (let the sum is $\partial$) using the binary search method discussed below in Section~\ref{sec:analysis}. 

\STATE Node $s$ checks the following  locally: 

\IF{$\partial < 4\eps $} \label{stp:checking1}
\STATE Output $\ell$ and  STOP.
\ENDIF


\ENDFOR  \label{stp:end-for-loop1}
\ENDFOR 

\end{algorithmic}
\end{algorithm}

The algorithm starts with the random walk length $\ell = 1$ and the computation proceeds in iterations. After each iteration, the value of $\ell$ is incremented by a factor $2$ i.e., doubled. In an iteration, the algorithm first computes the probability distribution $\pd_{\ell}$ of a random walk of length $\ell$ starting from the given source node $s$. For this, it uses Algorithm~\ref{alg:rw-probability} from the previous section. Then every node $u$ locally computes the difference $x_u = |p_{\ell}(u) - \beta/n|$ (the algorithm first looks for the minimum size mixing set $S$, i.e., of size $n/\beta$). The source node $s$ then collects $n/\beta$ smallest of those $x_u$s and checks if their sum is less than $4\eps$ (note that our algorithm will check for $4\eps$ instead of
$\eps$ for technical reasons that will be explained later). If `yes', then algorithm stops and outputs the length $\ell$ as the local mixing time. Otherwise, if the sum is greater than $4\eps$, the algorithm checks for the mixing set of size $(1 + \eps)n/\beta$. That is the source node collects $(1+\eps)n/\beta$ smallest of the differences $x_u = |p_{\ell}(u) - \beta/(1+\eps)n|$ and checks if their sum satisfies the $\eps$-condition.\footnote{It is shown in the analysis that we compute local mixing time with the accuracy parameter $4\eps$.} 
In general, if the sum of the $x_u$ values in a set $S$ did not satisfy the condition, the algorithm extends the search space by incrementing the size of the set by a factor of $(1+\eps)$. The algorithm starts with $|S| = n/\beta$ as the size of the local mixing set is at least $n/\beta$ (by the definition). This way the algorithm checks if there exists a set of size larger than $n/\beta$ where the random walk mixes locally. If such a set exists, the algorithm stops and outputs the length $\ell$ as the local mixing time. Else, the algorithm goes to the next iteration and does the same computation by doubling the random walk length to $2\ell$. The output of the algorithm is correct because it gives the existence of a set of size $\geq n/\beta$ where the local mixing time condition satisfies (cf. Definition~\ref{def:loc-mix-time}). The algorithm only computes the local mixing time and not the set where the random walk probability mixes. Hence, finding an $\ell$ satisfying the local mixing time condition is sufficient. The pseudocode is given in Algorithm~\ref{alg:local-mixing-time}. 
 

\subsection{Description and Analysis}\label{sec:analysis}

Let us now discuss the details of the computation in each iteration of Algorithm~\ref{alg:local-mixing-time}, where $\ell$ varies starting from $1$ and doubles in each iteration. \\

\noindent \textbf{Compute BFS tree from $s$:\footnote{Instead of computing a BFS tree in each iteration, one can simplify the algorithm by computing a BFS tree of depth $D$ just once in the beginning of the algorithm, i.e., before the for-loop on $\ell$. However, this will incur an additional $O(D)$ term in the running time of the algorithm.}} 
The source node $s$ computes a Breadth First Search (BFS) tree of depth $\mathcal{D} = \min\{D, \ell\}$ via flooding (see e.g., \cite{gopal-book}), where $D$ is the diameter of the graph. Each node knows its parent in the BFS tree. The BFS tree construction takes $O(\mathcal{D})$ rounds \cite{gopal-book}. \\

\noindent \textbf{Compute the probability distribution $\pd_{\ell}$ of a random walk of length $\ell$ starting from $s$:} The source node $s$ runs Algorithm~\ref{alg:rw-probability} with input $\ell$. At the end, each node $u$ will have the probability $p_{\ell}(u)$ (some of the $p_{\ell}(u)$s could be zero). This takes $O(\ell)$ rounds, see Section~\ref{sec:deterministic-simulation}. 

We next discuss the details of each iteration of the {\bf for} loop (steps~\ref{stp:doubling-set}-\ref{stp:end-for-loop1} of Algorithm~\ref{alg:local-mixing-time}) where the size of the set $R$ varies starting from $n/\beta$ 
and increases by a factor $1+\eps$ in each iteration. \\

\noindent \textbf{Every node $u$ computes the difference $x_u = |p_{\ell}(u) - 1/|S||$:} Since $|S| = n/\beta$ is known, each node $u$ can compute $x_u$ locally.  \\

\noindent \textbf{The source node $s$ collects $|S| = n/\beta$ smallest of $x_u$ values and checks if their sum is less than $4\eps$:} 
Each node $u$ sends its $x_u$ value to the source node $s$. A naive way of doing this is to upcast (see e.g., \cite{gopal-book}) all the values through the BFS tree edges in a pipelining manner. Then the source node $s$ can take the $|S|$ smallest of them and checks locally if the sum is less than $4\eps$. The upcast may take $\Omega(n)$ time in the worst case due to congestion in the BFS tree. 
\\
To overcome the congestion, we use the following efficient approach. Instead of collecting all the $x_u$ at $s$, the $|S|$ smallest of them can be found by doing a binary search on $\{x_u \,|\, u\in V\}$. All the nodes in the BFS tree send $x_{\min}$ and $x_{\max}$ (the minimum and maximum respectively among all $x_u$) to the root $s$ through a convergecast process (e.g., see \cite{gopal-book}). This will take time proportional to the depth of the BFS tree. Then $s$ can count the number of nodes whose $x_u$ value is less than $x_{\med} = (x_{\min} + x_{\max})/2$ via a couple of broadcast and convergecast. In fact, $s$ broadcasts the value $x_{\med}$ to all the nodes via the BFS tree and then the nodes whose $x_u$ value is less than $x_{\med}$ (say, the {\em qualified nodes}), reply back with $1$ value through the convergecast. Depending on whether the number of qualified nodes is less than or greater than $1/|S|$, the root updates the $x_{\med}$ value (by again collecting $x_{\min}$ or $x_{\max}$ in the reduced set) and iterates the process until the count is exactly $1/|S|$. Then $s$ can determine the sum of $x_u$s from the qualified nodes (by a convergecast) and checks locally if the sum is less than $4\eps$. As a summary, this is essentially finding $|S|$ smallest $x_u$ values through a binary search on all the $x_u$ (for all $u \in V$). Each broadcast and convergecast takes $O(\mathcal{D})$ time (more precisely, the depth of the BFS tree) and being done a constant number of times to compute size of the qualified set. Further, another $O(\log n)$ factor is incurred for the binary search over $x_u$s, which gives $O(\mathcal{D} \log n)$ time overall. 

There might be multiple nodes with the same  $x_u$ value. To handle this, each node $u$ chooses a very small random number $r_u$ and adds it to $x_u$ in the beginning. Then it can be shown that with high probability all the $(x_u + r_u)$ values are different and at the same time the addition does not affect the sum significantly (which has to be less than $4\eps$). For an example, say all the nodes choose a random number $r_u$ from the interval $[1/n^8, 1/n^4]$. Then by adding $r_u$ to all the $n$ nodes, the sum value will increase by at most $n\cdot 1/n^4 = 1/n^3$ which is much smaller than $\eps$. Further, using Chernoff's bound  it can be easily shown that with high probability the values $x_u+r_u$ are all distinct, since $r_u$s are distinct.  \\

\noindent \textbf{Incrementing the size of the local mixing set $S$ by a factor $(1+\eps)$:} In the first iteration,  algorithm  checks the local mixing on a set $S$ of size $n/\beta$. More specifically, the source node $s$ collects $|S| = n/\beta$ smallest of $x_u$ values and checks if their sum is less than $4\eps$. If true, then the algorithm stops and outputs $\ell$ as the local mixing time. If not, then the algorithm looks for larger set in the next iteration, i.e.,  size $(1+\eps)|S|$. The source node collects $(1+\eps) n/\beta$ smallest of $x_u$ values and checks if their sum is less than $4\eps$. If true, it outputs $\ell$; if not, then it checks on the incremented set of size $(1+\eps)^2|S|$ and so on. 
Below we discuss on why we check the sum condition with value $4\eps$ (cf. Lemma~\ref{lem:doubling-set-size}). 
The main idea behind the slightly relaxed  condition (i.e., $4\eps$) is that it indirectly checks whether the sum condition is satisfied,
i.e., $\sum |p_{\ell} - 1/s| < \eps$ for all set sizes {\em between the sizes that are actually checked}, i.e., $(1+\eps)^i n/\beta < s < (1+\eps)^{i+1} n/\beta$, for $i = 0, 1, \dots$.
In this way, for a particular length $\ell$, $s$ checks if there exists a local mixing set of size at least $n/\beta$. If $s$ is successful on some set, the algorithm stops and outputs the length $\ell$ as the local mixing time. Otherwise, if there is no such local mixing set (i.e., $\ell$ is not the local mixing time), the algorithm goes to the next iteration by doubling the length $\ell$ of the random walk. The output is correct because it gives the existence of a set where the local mixing time condition satisfies. Hence, finding an $\ell$ satisfying the local mixing time condition is sufficient. The following lemma shows the correctness of the above incrementation approach. 

\begin{lemma}\label{lem:doubling-set-size}
Let $S_1$  be any set (of smallest $x_u$ values) such that $|S_1|$ lies between $|S|$ and $(1+\eps)|S|$, i.e.,
$|S| < |S_1| < (1+\eps)|S|$. Let $S_2 \supseteq S_1$ be the set of size $(1+\eps)|S|$ (this is a set considered by
the algorithm). Further assume that $\sum_{u \in S_1} |p_{\ell}(u) - \frac{1}{|S_1|}| < \eps$. Then
$ \sum_{u \in S_2} \left|p_{\ell}(u) - \frac{1}{(1+\eps)|S|}\right| < 4\eps.$
\end{lemma} 
\begin{proof}
We have, $|S_2 \setminus S_1| = \eps |S|$. First note that:
\begin{align}
\left|p_{\ell}(u) - \frac{1}{(1+\eps)|S|}\right| & \leq \left|p_{\ell} (u) - \frac{1}{|S_1|}\right| + \left|\frac{1}{|S_1|} - \frac{1}{(1+\eps)|S|}\right| \nonumber \\
& \leq \left|p_{\ell} (u) - \frac{1}{|S_1|}\right| + \left|\frac{1}{|S|} - \frac{1}{(1+\eps)|S|}\right| \nonumber \\
& = \left|p_{\ell} (u) - \frac{1}{|S_1|}\right| + \frac{1}{|S|}\frac{\eps}{1+\eps} \nonumber
\end{align}
Therefore, 
\begin{align}
\sum_{u \in S_1} \left|p_{\ell}(u) - \frac{1}{(1+\eps)|S|}\right| & \leq \sum_{u \in S_1} \left|p_{\ell}(u) - \frac{1}{|S_1|}\right| + \frac{\eps |S_1|}{(1+\eps)|S|}  \nonumber \\ 
& <  2\eps \label{eq:one}
\end{align}
Also note that: 
\begin{align*}
& \sum_{u \in S_1} \left|p_{\ell}(u) - \frac{1}{|S_1|}\right| < \eps \\
& \Rightarrow -\eps < \sum_{u \in S_1} \left(p_{\ell}(u) - \frac{1}{|S_1|}\right) < \eps  \Rightarrow \sum_{u \in S_1} p_{\ell}(u) > 1 - \eps \\
& \Rightarrow \sum_{u \in S_2\setminus S_1} p_{\ell}(u) \leq \eps  \hspace{1cm} \text{[since, $\sum_{u \in V} p_{\ell}(u) = 1$]}
\end{align*}
Then,  
\begin{align}
\sum_{u \in S_2\setminus S_1} \left|p_{\ell}(u) - \frac{1}{(1+\eps)|S|}\right| & \leq \sum_{u \in S_2\setminus S_1} p_{\ell}(u) + \frac{|S_2\setminus S_1|}{(1+\eps)|S|}  \nonumber \\ 
& = \sum_{u \in S_2\setminus S_1} p_{\ell}(u) + \frac{\eps |S|}{(1+\eps)|S|}  \nonumber \\ 
& < \eps + \eps = 2\eps \label{eq:two}
\end{align}
Furthermore, since the algorithm compares the sum of the smallest differences in all the sets, we get,
\begin{align*}
\sum_{u \in S_2} \left|p_{\ell}(u) - \frac{1}{(1+\eps)|S|}\right| & = \sum_{u \in S_1} \left|p_{\ell}(u) - \frac{1}{(1+\eps)|S|}\right| \\ & + \sum_{u \in S_2\setminus S_1} \left|p_{\ell}(u) - \frac{1}{(1+\eps)|S|}\right| \\
& < 4\eps \hspace{.5cm}\text{[from Equation~\ref{eq:one} and ~\ref{eq:two}]}
\end{align*}
\end{proof}

The above lemma says that if there is a set $S_1$ of size that lies  between $|S|$  and $(1+\eps)|S|$ such that the sum difference in $S_1$ is less than $\eps$, then the sum difference in the incremented set $S_2$ of size $(1+\eps)|S|$ is less than $4\eps$. Moreover, if the sum $\sum_{u \in S_1} |p_{\ell}(u) - \frac{1}{|S_1|}| \geq 4\eps$, then automatically the sum is greater than $\eps$.  Hence, it is sufficient to check with $4\eps$ for all the sets of size $\geq n/\beta$. \\

\noindent \textbf{Doubling the length $\ell$ after each iteration:} Finally, we show that the doubling of the random walk length $\ell$ in each iteration gives a $2$-approximation of the local mixing time $\tau_s(\beta, \eps)$. We remark that the monotonicity property of the distribution $\pd_{\ell}$ doesn't hold over a restricted set $S\subset V$ in general. Thus the local mixing time is not monotonic, unlike the mixing time of a graph, see Lemma~\ref{lem:monotonicity}. Hence, in general, binary search on length $\ell$ will not work. 
However, the idea of doubling the length $\ell$ in each iteration will work as we show that the amount of probability that goes out from a set $S$ (where the walk mixes locally) in the next $\ell$ steps of the random walk is very small i.e., $o(1)$. As we discussed in Section~\ref{sec:graph-examples}, the local mixing time is interesting and effective on the graphs where the local mixing time is very small compared to the mixing time. Also the mixing time estimates the conductance of the graph. This intuitively justifies our assumption $\tau_s(\beta, \eps)\phi(S) = o(1)$, where $\tau_s(\beta, \eps)$ is the local mixing time w.r.t the source $s$ and $\phi(S)$ is the conductance of the set $S$ ($S$ is the set where the random walk mixes locally). Recall that the conductance of the set $S$ is defined as $\phi(S) = |E(S, V\setminus S)|/\min\{\mu(S), \mu(V\setminus S)\}$.     


Suppose $\ell$ be the local mixing time and $S$ is the set where the random walk locally mixes. Then we show that starting from the stationary distribution in $S$, the  amount of probability   that goes out of the set $S$ after another $\ell$ steps of the walk is at most $\ell \phi(S)$. 


\begin{lemma}\label{sec:escaping-mass}
Let $S\subset V$ be a set of size $n/\beta$ where a random walk probability distribution locally mixes in $\ell$ steps when started from a source node $s\in S$. Let $\pd_{\ell}$ 
be the probability distribution at time $\ell = \tau_s(\beta, \eps)$. Assume, $\tau_s(\beta, \eps)\phi(S) = o(1)$. Then  $||\pd_{2\ell}{\restriction_S} - 1/|S|||_1 < 2\eps$, i.e.,
the local mixing time condition in $S$ is satisfied (with parameter $2\eps$) at length $2\ell$.
\end{lemma}    
\begin{proof}
 Since  $\ell$ is the local mixing time, the restricted probability distribution $\pd_{\ell}{\restriction_S}$ is $\eps$-close to the stationary distribution in $S$. 
  Let $E(S, V\setminus S)$ be the set of edges between $S$ and  $V\setminus S$.  
  The amount of probability goes out of $S$ in one step (i.e., at time $\ell + 1$) is $|E(S, V\setminus S)|/d|S|$ (each crossing edge carries $1/d|S|$ fraction of the probability since the graph is $d$-regular). Note that some amount of probability may come in to $S$, but that's good for our upper bound claim. We know that conductance of $S$ is $\phi(S) = |E(S, V\setminus S)|/d|S|$.
 Therefore, the total amount of probability that goes out of $S$ in the next $\ell$ steps (i.e., at time $2\ell$) is at most $\ell \phi(S)$.
Thus, $||\pd_{2\ell}{\restriction_S} - \pd_{\ell}{\restriction_S}||_1 \leq \ell \phi(S) = o(1)$.   
Hence, it follows from the assumption $\tau_s(\beta, \eps)\phi(S) = o(1)$ that the amount of probability that goes out of the set $S$ is $o(1)$.
 Moreover, $||\pd_{2\ell}{\restriction_S} - 1/|S|||_1 \leq ||\pd_{2\ell}{\restriction_S} - \pd_{\ell}{\restriction_S}||_1 + ||\pd_{\ell}{\restriction_S} - 1/|S|||_1 \leq \ell \phi(S) + ||\pd_{\ell}{\restriction_S} - 1/|S|||_1$. Hence for $\ell = \tau_s(\beta, \eps)$, $||\pd_{2\ell}{\restriction_S} - 1/|S|||_1 < \eps + \eps = 2\eps$, since $o(1) \leq \eps$ and  $||\pd_{\ell}{\restriction_S} - 1/|S|||_1 < \eps$. That is at length $2\ell$,  the local mixing time condition in $S$ is satisfied (with parameter $2\eps$).
 \end{proof}
 
 From the above lemma, it follows that the {\sc Local-Mixing-Time} algorithm must stop by the time $2 \tau_s(\beta, \eps)$, if the algorithm misses the exact local mixing time $\tau_s(\beta, \eps)$ when doubling the length in an iteration\footnote{Recall that the algorithm checks the $L_1$-norm condition with the accuracy parameter $4\eps$. Hence, it subsumes the $2\eps$ case when the length is doubling.}. Therefore, the output length $\ell$ is at most a $2$-approximation of the local mixing time. 

The running time of the above algorithm to compute the local mixing time is given in the following theorem. 
\begin{theorem}\label{thm:main-time-result}
Given an undirected regular graph $G$, a source node $s$ and a positive integer $\beta$, the algorithm {\sc Local-Mixing-Time} computes  a $2$-approximation of the local mixing time $\tau_{s}(\beta, \eps)$ with high probability and finishes in $O(\tau_s(\beta, \eps) \log^2 n \log_{(1+\eps)} \beta)$ time, provided $\tau_{s}(\beta, \eps) \phi(S) = o(1)$, where $S$ is the local mixing set.  
\end{theorem}
\begin{proof} 
The correctness of the algorithm is described above. We calculate the running time. 
The algorithm iterates $O(\log n)$ times, for $\ell = 1, 2, 4, \ldots, 2^i, \ldots, 2\tau_s(\beta, \eps)$. In each iteration: 
\begin{enumerate}[(1)]
\item the source node computes a BFS tree, which takes $O(\mathcal{D})$ rounds. 
\item the algorithm runs Algorithm~\ref{alg:rw-probability} as a subroutine. It takes $O(\ell)$ rounds. 
\item $s$ collects the sum of $R$ smallest $x_u$s through the BFS tree. It takes $O(\mathcal{D}\log n)$ rounds. This is done for all $R = (1+\eps)^i|S|$, where $i=0, 1, 2, \ldots$. It will take $O(\log_{(1+\eps)} \beta)$ time as  $(1+\eps)^i|S| = n$. Hence the time taken is $O(\mathcal{D}\log n \log_{(1+\eps)} \beta)$ rounds.
\item  checking if the sum of differences is less than $\eps$ and $4\eps$, can be done locally at $s$.
\end{enumerate} 
Since $\ell \leq  2\tau_s(\beta, \eps)$, the total time required is $(O(\mathcal{D}) + O(\tau_s(\beta, \eps)) + O(\mathcal{D} \log n \log_{(1+\eps)} \beta))\log n$, which is bounded by $O(\tau_s(\beta, \eps) \log^2 n \log_{(1+\eps)} \beta)$, since $\mathcal{D} \leq \ell \leq 2\tau_s(\beta, \eps)$.
\end{proof}

\subsection{Algorithm for Computing Exact Local Mixing Time}\label{sec:general-graph-local-mixing-time}

The above algorithm finds a $2$-approximation of the local mixing time. It can be extended to compute {\em exact} local mixing time corresponding to the given parameters $\beta$ and $\eps$. Moreover, the extended algorithm works for any general regular graph. The running time of the extended algorithm will increase by a $\mathcal{D}$ multiplicative factor than the running time of the previous $2$-approximation algorithm (Algorithm~\ref{alg:local-mixing-time}). The extended algorithm follows the same internal steps as in Algorithm~\ref{alg:local-mixing-time}, except the number of iterations. Instead of doubling the length $\ell$ in each iteration, the algorithm iterates for each value of $\ell = 1, 2, 3, \ldots$. The algorithm starts with $\ell = 1$ and the computation proceeds in iterations. In each iteration, the algorithm runs Steps~\ref{stp:bfs-tree}-\ref{stp:end-for-loop1} of Algorithm~\ref{alg:local-mixing-time}.  

Now we explain how to compute the probability distribution $\pd_{\ell}$ of a random walk of length $\ell$ from the previous distribution $\pd_{\ell-1}$ in {\em one} round. We resume the deterministic flooding technique from the last step $(\ell - 1)$ with the probability distribution $\pd_{\ell-1}$ and compute $\pd_{\ell}$ in one step by flooding. In particular, starting from the distribution $\pd_{\ell-1}$, the algorithm runs Step~\ref{stp:det-flooding} of Algorithm~\ref{alg:rw-probability}. The Step~\ref{stp:det-flooding} essentially computes the probability distribution $\pd_{\ell}$ from the distribution $\pd_{\ell-1}$ in one round. 

 Since for each length $\ell = 1, 2, \dots$, the source node checks if there exists a set $S$ where the probability mixes, the algorithm will find the exact local mixing time $\tau_s(\beta, \eps)$. At the same time, the algorithm works for an arbitrary regular graph without the condition $\tau_s(\beta, \eps)\phi(S) =o(1)$ (since we are not doubling the length). The running time of the algorithm to compute exact $\tau_{s}(\beta, \eps)$ is given in the following theorem.

\begin{theorem}
Suppose $\tau_{s}(\beta, \eps)$ is the local mixing time w.r.t. the vertex $s$. There is an algorithm which computes  $\tau_{s}(\beta, \eps)$ with high probability and finishes in $O(\tau_{s}(\beta, \eps) \mathcal{D} \log n \log_{(1+\eps)} \beta)$ time. 
\end{theorem}

\begin{proof} 
The algorithm iterates for each $\ell = 1, 2, \dots, \tau_{s}(\beta, \eps)$. Inside each iteration, first the source node $s$ computes a BFS tree, which takes $O(\mathcal{D})$ rounds, then it runs Step~\ref{stp:det-flooding} of Algorithm~\ref{alg:rw-probability} and then all the other steps of Algorithm~\ref{alg:local-mixing-time}. Therefore, inside one iteration, the total time taken is $(O(\mathcal{D}) + O(1) + O(\mathcal{D} \log n \log_{(1+\eps)} \beta))$, which is bounded by $O(\mathcal{D} \log n \log_{(1+\eps)} \beta)$ (cf. Theorem~\ref{thm:main-time-result}). 
Since the number of iterations is $\tau_s(\beta, \eps)$, the time complexity of the algorithm is  $O(\tau_{s}(\beta, \eps) \mathcal{D} \log n \log_{(1+\eps)} \beta)$. Recall that $\mathcal{D}$ which is bounded above by $D$, could be much smaller than $\tau_s(\beta, \eps)$.
\end{proof}

\section{Application to Partial Information Spreading}
\label{sec:app}
A main application of local mixing is that the local mixing time characterizes  partial information spreading.
 As mentioned in Section \ref{sec:intro}, partial information spreading has many applications including to the maximum coverage problem \cite{censor-podc10} and full information spreading \cite{censor-soda11}. 

The partial information spreading problem is defined in \cite{censor-podc10} which can be considered  a relaxed version of the
well-studied (full) information spreading problem (see e.g., \cite{soda13, tcs15}). Initially each node has a token, and unlike the full information spreading (which requires to send each token to all the nodes), the requirement of  partial information spreading is to send each token to only $n/\beta$ nodes, and every node should have $n/\beta$ different tokens.  A formal definition is:

\begin{definition}\label{def:partial-info-spred}
 Each node $v\in V$ has a token $m(v)$. For a given constant $\beta \geq 1$ and for any fixed $\delta \in(0, 1)$, a $(\delta, \beta)$-partial information spreading means  that, with probability at least $1 - \delta$, each token $m(v)$ disseminates to at least $n/\beta$ nodes and every node receives at least $n/\beta$ different tokens.  
\end{definition}

To study partial information spreading, we use the well-studied (synchronous) push/pull model of communication, where each node chooses, in each round, a random neighbor to exchange information with. Note that this algorithm assumes the LOCAL model, i.e.,
in each round, there is no limit on the number of messages (tokens) that can be exchanged over an edge.
In this setting, information spreading and partial information algorithms under the push-pull model have been extensively studied (see, e.g., \cite{censor-podc10, censor-soda11} and the references therein).
 
We show that  partial information spreading in regular graphs can be accomplished 
in $\tilde{O}(\tau(\beta, \eps))$ rounds with high probability,  where $\tau(\beta, \eps)$ is the local mixing time of the graph. 
We show that it holds in the LOCAL model (where there is no congestion); as mentioned earlier, LOCAL model is typically used
in prior literature to analyze the 
push-pull mechanism \cite{censor-podc10, censor-soda11}.\footnote{In the CONGEST model, the bound will be  $\tilde{O}(\tau(\beta, \eps)  + n/\beta)$; note that $\Omega(n/\beta)$ is a time lower bound in general, since a node with degree $d$ needs at least $\Omega(\frac{n}{\beta d})$ rounds to get  $n/\beta$ different tokens.}
We compare our bound with the previous bound of $O(\frac{\log n + \log (1/\delta)}{\Phi_{\beta}(G)})$ of \cite{censor-podc10} for $(\delta, \beta)$-partial information spreading, where $\Phi_{\beta}(G)$ is the weak conductance of the graph. (Note that this bound
is also for the LOCAL model.)  This  bound is for the  ``push-pull" algorithm: in every round, each node $i$ chooses a random neighbor $j$ and exchanges  information (all their respective tokens)  with it. Note that the algorithm does not specify any termination condition (i.e., how long
show it run). To specify  that, one should know a bound on
the weak conductance (which is not known a priori). In contrast, we show that local mixing time (also) characterizes the run time of partial information spreading and our distributed computation of local mixing time (in the previous section) helps us to specify a termination condition for the push-pull mechanism.

We note that our bounds based on local mixing time are comparable to the bound based on weak conductance in many graphs (in fact, we conjecture a tight relationship between local mixing time and weak conductance, in the manner similar to the relationship
between mixing time and conductance).
However, the analysis of our bound is quite different to the one that uses weak conductance; it is simpler to analyze using random walks.
We show our bound in the LOCAL model, which can be easily extended  to the CONGEST model.


\begin{theorem}
Partial information spreading in any (regular) graph can be accomplished by running the ``push-pull" algorithm for $\tilde{O}(\tau(\beta, \eps))$ rounds 
with high probability (whp), i.e., with probability at least $1 -1/n^c$ for some constant $c >0$.
\end{theorem}
\begin{proof} (sketch)
First, we show that every token is disseminated to at least $n/\beta$ nodes whp. 
To analyze the performance of push-pull, it is enough to focus on a single message (token) and bounding the time taken
for the message to reach at least $n/\beta$ nodes whp. Then, using union bound, it will follow that every token reaches at least $n/\beta$ nodes whp.

Fix a token $a$. Let the token be initially at node $v_1$. 
 The analysis proceeds in phases with each phase
consisting of $\tau(\beta, \eps)$  (i.e., equal to the local mixing time) rounds.  Since the token $a$ propagates
by push-pull, this is equivalent to the token performing a random walk starting from node $v_1$ (in each round, a random neighbor 
is chosen and token is sent to that neighbor).
Thus, in every phase, a random walk of length $\tau(\beta, \eps)$ is performed. The first phase starts by performing a random walk of length $\tau(\beta, \eps)$ from $v_1$.  Suppose the random walk ends at $u_1$ (can be same as $v_1$, but most likely different, as will be seen). Then in the second phase, there are two source nodes, namely $v_1$ and $u_1$ (we even ignore the dissemination of token $a$ by the intermediate nodes on the random walk path). Suppose the two random walks from $v_1$ and $u_1$, end at (potentially different)  nodes  $v_2$ and $u_2$. Then in the next phase, there are potentially four source nodes (namely,  $v_1$, $u_1$, $v_2$, $u_2$) from where random walk of length $\tau(\beta, \eps)$ each are performed in parallel.

We claim that after $O(\log n)$ phases, the token has reached at least $n/\beta$ nodes. This can be shown by using a standard coupon collector argument. Let the local mixing set (cf. Definition \ref{def:loc-mix-time}) corresponding to the source node $v_1$ is $S_1$. Then, by the definition of the local mixing time, the size of $S_1$ is at least $n/\beta$. Thus, in the first phase, token $a$ ends up (almost) uniformly at random in the set $S_1$. If the token ends up at $u_1$, and let $S_2$ be the local mixing set corresponding
to the source node $u_1$ ($S_1$ and $S_2$ might be the same, it does not matter). At the end of the second phase, the  random walks
starting from $v_1$ and $u_1$ will end up uniformly at random in the respective sets $S_1$ and $S_2$ (note that
the local mixing time of the graph, $\tau(\beta, \eps)$  is the maximum among the local mixing times with respect to the all the 
$n$ different source nodes). Hence, by a coupon collector argument, it follows that after $O(\tau(\beta, \eps) \log n)$ rounds
token $a$ would have reached at least $n/\beta$ nodes whp. By union bound, the same applies to all tokens whp.

To argue that each node receives at least $n/\beta$ tokens (whp), we use the fact that in a regular graph, the random walk
is symmetric (and reversible) and hence one can apply the argument in reverse and show that (whp) at least $n/\beta$ tokens
is received by any fixed node whp;  an union bound shows the same for all nodes.
\end{proof}


\section{Conclusion}
We introduced the  notion of local mixing time and presented an efficient random-walk based distributed algorithm with provable guarantees to compute the same in undirected (regular) graphs. Our algorithm is simple and lightweight, and estimates the local mixing time with high accuracy. We showed that local mixing time can be used to characterize partial information spreading (which has many applications) and, indeed, it can be used as a termination condition in the push-pull based partial information spreading algorithm. 

Several open problems arise from our work. One key problem is whether it is possible to compute the local mixing time efficiently,
i.e., in $\tilde{O}(\tau(\beta, \eps))$ rounds in arbitrary graphs. Finding a relationship between local mixing time and weak conductance is another key problem.


\bibliographystyle{abbrv}
\bibliography{Distributed-RW.bib}

\begin{thebibliography}{10}

\bibitem{storage-spaa13}
J.~Augustine, A.~R. Molla, E.~Morsy, G.~Pandurangan, P.~Robinson, and E.~Upfal.
\newblock Search and storage in dynamic peer-to-peer networks.
\newblock In {\em Proc. of 25th ACM Symposium on Parallelism in Algorithms and
  Architectures (SPAA)}, pages 53--62, 2013.

\bibitem{APR-podc13}
J.~Augustine, G.~Pandurangan, and P.~Robinson.
\newblock Fast byzantine agreement in dynamic networks.
\newblock In {\em Proc. of 32nd Annual ACM SIGACT-SIGOPS Symposium on
  Principles of Distributed Computing (PODC)}, pages 74--83, 2013.

\bibitem{AFPP-soda12}
V.~Auletta, D.~Ferraioli, F.~Pasquale, and G.~Persiano.
\newblock Metastability of logit dynamics for coordination games.
\newblock In {\em Proc of the 23rd Annual {ACM-SIAM} Symposium on Discrete
  Algorithms (SODA)}, pages 1006--1024, 2012.

\bibitem{censor-podc10}
K.~Censor{-}Hillel and H.~Shachnai.
\newblock Partial information spreading with application to distributed maximum
  coverage.
\newblock In {\em Proc. of the 29th Annual {ACM} Symposium on Principles of
  Distributed Computing (PODC)}, pages 161--170, 2010.

\bibitem{censor-soda11}
K.~Censor{-}Hillel and H.~Shachnai.
\newblock Fast information spreading in graphs with large weak conductance.
\newblock In {\em Proc of 22nd Annual {ACM-SIAM} Symposium on Discrete
  Algorithms (SODA)}, pages 440--448, 2011.

\bibitem{DasSarmaGP09}
A.~{Das Sarma}, S.~Gollapudi, and R.~Panigrahy.
\newblock Sparse cut projections in graph streams.
\newblock In {\em Proc. of 17th Annual European Symposium on Algorithms (ESA)},
  pages 480--491, 2009.

\bibitem{tcs15}
A.~{Das Sarma}, A.~R. Molla, and G.~Pandurangan.
\newblock Distributed computation in dynamic networks via random walks.
\newblock {\em Theor. Comput. Sci.}, 581:45--66, 2015.

\bibitem{SarmaMP15}
A.~{Das Sarma}, A.~R. Molla, and G.~Pandurangan.
\newblock Distributed computation of sparse cuts via random walks.
\newblock In {\em Proc. of 16th International Conference on Distributed
  Computing and Networking (ICDCN)}, pages 6:1--6:10, 2015.

\bibitem{DasSarmaMPU15}
A.~{Das Sarma}, A.~R. Molla, G.~Pandurangan, and E.~Upfal.
\newblock Fast distributed pagerank computation.
\newblock {\em Theor. Comput. Sci.}, 561:113--121, 2015.

\bibitem{drw-jacm}
A.~{Das Sarma}, D.~Nanongkai, G.~Pandurangan, and P.~Tetali.
\newblock Distributed random walks.
\newblock {\em J. ACM}, 60(1):2, 2013.

\bibitem{DM15}
P.~Diaconis and L.~Miclo.
\newblock On quantitative convergence to quasi-stationarity.
\newblock {\em Ann. Fac. Sci. Toulouse Math.Sér. 6}, 24:973--1016, 2015.

\bibitem{soda13}
C.~Dutta, G.~Pandurangan, R.~Rajaraman, Z.~Sun, and E.~Viola.
\newblock On the complexity of information spreading in dynamic networks.
\newblock In {\em Proceedings of the Twenty-Fourth Annual {ACM-SIAM} Symposium
  on Discrete Algorithms, {SODA}}, pages 717--736, 2013.

\bibitem{mihail}
C.~Gkantsidis, M.~Mihail, and A.~Saberi.
\newblock Throughput and congestion in power-law graphs.
\newblock In {\em SIGMETRICS}, pages 148--159, 2003.

\bibitem{JS89}
M.~Jerrum and A.~Sinclair.
\newblock Approximating the permanent.
\newblock {\em SIAM Journal of Computing}, 18(6):1149--1178, 1989.

\bibitem{kempe}
D.~Kempe and F.~McSherry.
\newblock A decentralized algorithm for spectral analysis.
\newblock {\em Journal of Computer and System Sciences}, 74(1):70--83, 2008.

\bibitem{KM15}
F.~Kuhn and A.~R. Molla.
\newblock {Distributed Sparse Cut Approximation}.
\newblock In {\em Proc. of 19th International Conference on Principles of
  Distributed Systems (OPODIS)}, pages 10:1--10:14, 2015.

\bibitem{Levin}
D.~A. Levin, Y.~Peres, and E.~L. Wilmer.
\newblock {\em Markov Chains and Mixing times}.
\newblock American Mathematical Society, Providence, RI, USA, 2008.

\bibitem{icdcn17}
A.~R. Molla and G.~Pandurangan.
\newblock Distributed computation of mixing time.
\newblock In {\em Proc. of 18th International Conference on Distributed
  Computing and Networking (ICDCN)}, page~5, 2017.

\bibitem{sirocco14}
G.~Pandurangan.
\newblock Distributed algorithmic foundations of dynamic networks.
\newblock In {\em Proc. of 21st International Colloquium on Structural
  Information and Communication Complexity (SIROCCO)}, pages 18--22, 2014.

\bibitem{gopal-book}
G.~Pandurangan.
\newblock {\em Distributed Network Algorithms}.
\newblock https://sites.google.com/site/gopalpandurangan/dna, 2016.

\bibitem{peleg}
D.~Peleg.
\newblock {\em Distributed computing: a locality-sensitive approach}.
\newblock SIAM, Philadelphia, PA, USA, 2000.

\bibitem{SpielmanT04}
D.~A. Spielman and S.~Teng.
\newblock Nearly-linear time algorithms for graph partitioning, graph
  sparsification, and solving linear systems.
\newblock In {\em Proc. of 36th ACM Symposium on Theory of Computing (STOC)},
  pages 81--90, 2004.

\end{thebibliography}

\end{document}